\documentclass[conference,a4paper]{IEEEtran}

\usepackage[utf8]{inputenc} 
\usepackage[T1]{fontenc}
\usepackage{url}
\usepackage{ifthen}
\usepackage{cite}
\usepackage[cmex10]{amsmath} 

\usepackage{amsmath}
\usepackage{amsthm}
\usepackage{amssymb}
\usepackage{multirow}
\usepackage{amsfonts}
\usepackage{graphicx}
\usepackage{algorithm}
\usepackage{url}
\usepackage{lipsum}

\usepackage{latexsym}
\usepackage{color}
\usepackage{graphicx}
\usepackage{threeparttable}
\usepackage{float}

\usepackage{algpseudocode}
\usepackage{etoolbox}
\usepackage{tikz}
\usetikzlibrary{tikzmark}
\usetikzlibrary{calc}

\errorcontextlines\maxdimen

\newcommand{\ALGtikzmarkcolor}{black}
\newcommand{\ALGtikzmarkextraindent}{4pt}
\newcommand{\ALGtikzmarkverticaloffsetstart}{-.5ex}
\newcommand{\ALGtikzmarkverticaloffsetend}{-.5ex}
\makeatletter
\newcounter{ALG@tikzmark@tempcnta}

\newcommand\ALG@tikzmark@start{%
    \global\let\ALG@tikzmark@last\ALG@tikzmark@starttext%
    \expandafter\edef\csname ALG@tikzmark@\theALG@nested\endcsname{\theALG@tikzmark@tempcnta}%
    \tikzmark{ALG@tikzmark@start@\csname ALG@tikzmark@\theALG@nested\endcsname}%
    \addtocounter{ALG@tikzmark@tempcnta}{1}%
}

\def\ALG@tikzmark@starttext{start}
\newcommand\ALG@tikzmark@end{%
    \ifx\ALG@tikzmark@last\ALG@tikzmark@starttext
    \else
        \tikzmark{ALG@tikzmark@end@\csname ALG@tikzmark@\theALG@nested\endcsname}%
        \tikz[overlay,remember picture] \draw[\ALGtikzmarkcolor] let \p{S}=($(pic cs:ALG@tikzmark@start@\csname ALG@tikzmark@\theALG@nested\endcsname)+(\ALGtikzmarkextraindent,\ALGtikzmarkverticaloffsetstart)$), \p{E}=($(pic cs:ALG@tikzmark@end@\csname ALG@tikzmark@\theALG@nested\endcsname)+(\ALGtikzmarkextraindent,\ALGtikzmarkverticaloffsetend)$) in (\x{S},\y{S})--(\x{S},\y{E});%
    \fi
    \gdef\ALG@tikzmark@last{end}%
}

\apptocmd{\ALG@beginblock}{\ALG@tikzmark@start}{}{\errmessage{failed to patch}}
\pretocmd{\ALG@endblock}{\ALG@tikzmark@end}{}{\errmessage{failed to patch}}
\algblockdefx[Foreach]{Foreach}{EndForeach}[1]{\textbf{foreach} #1 \textbf{do}}{\textbf{end foreach}}

\DeclareGraphicsExtensions{.tif,.pdf,.jpeg,.png,.jpg,.bmp,.svg,.eps,.EMF}

\newtheorem{theorem}{Theorem}
\newtheorem{definition}{Definition}
\newtheorem{lemma}{Lemma}
\newtheorem{corollary}{Corollary}

\newcommand{\cM}{\mathcal{M}}
\newcommand{\cA}{\mathcal{A}}

\newcommand{\cT}{\mathcal{T}}

\newcommand{\cG}{\mathcal{G}}

\newcommand{\bX}{\mathbf{x}}
\newcommand{\bY}{\mathbf{y}}

\newcommand{\bZ}{\mathbf{z}}

\newcommand{\poly}{\mathsf{poly}}
\newcommand{\dec}{\mathrm{dec}}
\newcommand{\decode}{\mathsf{decode}}
\newcommand{\diag}{\mathsf{diag}}
\newcommand{\add}{\mathsf{add}}
\newcommand{\supp}{\mathsf{supp}}

\newcommand{\sfA}{\mathsf{A}}
\newcommand{\sfB}{\mathsf{B}}
\newcommand{\Wx}{\mathsf{W}(x)}

\begin{document}

\title{Improved encoding and decoding for non-adaptive threshold group testing\\[.2ex] 
  {\normalfont\large 
	Thach V. Bui\IEEEauthorrefmark{1}, Minoru Kuribayashi\IEEEauthorrefmark{3}, Mahdi Cheraghchi\IEEEauthorrefmark{4}, and Isao Echizen\IEEEauthorrefmark{1}\IEEEauthorrefmark{2}}\\[-1.45ex]}

\author{\IEEEauthorblockA{\IEEEauthorrefmark{1}SOKENDAI (The Graduate\\University for Advanced \\Studies), Kanagawa, Japan\\ bvthach@nii.ac.jp}
\and
\IEEEauthorblockA{\IEEEauthorrefmark{3}Graduate School of Natural \\Science and Technology, \\Okayama University, Japan\\kminoru@okayama-u.ac.jp}
\and
\IEEEauthorblockA{\IEEEauthorrefmark{4}Department of Computing,\\Imperial College London, UK\\m.cheraghchi@imperial.ac.uk}
\and
\IEEEauthorblockA{\IEEEauthorrefmark{2}National Institute\\ of Informatics, \\Tokyo, Japan \\ iechizen@nii.ac.jp}}

\maketitle

\thispagestyle{plain}
\pagestyle{plain}

\begin{abstract}
The goal of threshold group testing is to identify up to $d$ defective items among a population of $n$ items, where $d$ is usually much smaller than $n$. A test is positive if it has at least $u$ defective items and negative otherwise. Our objective is to identify defective items in sublinear time the number of items, e.g., $\poly(d, \ln{n}),$ by using the number of tests as low as possible. In this paper, we reduce the number of tests to $O \left( h \times \frac{d^2 \ln^2{n}}{\mathsf{W}^2(d \ln{n})} \right)$ and the decoding time to $O \left( \dec_0 \times h \right),$ where $\dec_0 = O \left( \frac{d^{3.57} \ln^{6.26}{n}}{\mathsf{W}^{6.26}(d \ln{n})} \right) + O \left( \frac{d^6 \ln^4{n}}{\mathsf{W}^4(d \ln{n})} \right)$, $h = O\left( \frac{d_0^2 \ln{\frac{n}{d_0}}}{(1-p)^2} \right)$ , $d_0 = \max\{u, d - u \}$, $p \in [0, 1),$ and $\Wx = \Theta \left( \ln{x} - \ln{\ln{x}} \right).$ If the number of tests is increased to $O\left( h \times \frac{d^2\ln^3{n}}{\mathsf{W}^2(d \ln{n})} \right),$ the decoding complexity is reduced to $O \left(\dec_1 \times h \right),$ where $\dec_1 = \max \left\{ \frac{d^2 \ln^3{n}}{\mathsf{W}^2(d \ln{n})}, \frac{ud \ln^4{n}}{\mathsf{W}^3(d \ln{n})} \right\}.$ Moreover, our proposed scheme is capable of handling errors in test outcomes.
\end{abstract}

\section{Introduction}
\label{sec:intro}

Detection of up to $d$ defective items in a large population of $n$ items is the main objective of group testing proposed by Dorfman~\cite{dorfman1943detection}. In this seminal work, instead of testing each item one by one, he proposed to pool a group of items for reducing the number of tests. In classical group testing (CGT), the outcome of a test on a subset of items is positive if the subset has at least one defective item and negative otherwise. Damaschke~\cite{damaschke2006threshold} introduced \textit{threshold group testing} (TGT) by revising the definition of the outcome of a test. The outcome of a test on a subset of items is positive if the subset has at least $u$ defective items, negative if it has up to $\ell$ defective items, where $0 \leq \ell < u$, and arbitrary otherwise. The parameter $g = u - \ell - 1$ is called the gap. When $g = 0$, i.e. $\ell = u - 1$, threshold group testing has no gap. When $u = 1$, TGT reduces to CGT. Threshold group testing can be consider as a special case of complex group testing~\cite{chen2008upper} or generalized group testing with inhibitors~\cite{bui2018framework}. Most of existing work, such as~\cite{damaschke2006threshold} and~\cite{cheraghchi2013improved,de2017subquadratic,chan2013stochastic,chen2009nonadaptive}, deal with $g \geq 0.$ In this paper, the focus is on threshold group testing with no gap, i.e., $g = 0.$

There are two fashions for designing tests. The first is \textit{adaptive group testing} in which the design of a test depends on the designs of the previous tests. This fashion usually consumes lots of time for implementing, however, achieves optimal bounds on the number of tests. For tackling with slow implementing time, \textit{non-adaptive group testing} (NAGT) is introduced. In this fashion, all tests are designed in a priori and performed simultaneously. Nowadays, NAGT is widely applied in several applications such as computational and molecular biology~\cite{du2000combinatorial}, multiple access channels~\cite{d2018separable}, and neuroscience~\cite{bui2018framework}. The focus of this work is on NAGT. The term CNAGT stand for Classical Non-Adaptive Group Testing which is CGT associated with NAGT. Similarly, the term NATGT stands for Non-Adaptive Threshold Group Testing, which is TGT associated with NAGT. When there is no gap, NATGT is denoted as $u$-NATGT.

In any model of group testing, it is enticing to minimize the number of tests and to efficiently identify the set of defective items. CGT has been intensively studied for a long time for resolving these two requirements. By using AGT, the number of tests is $\Omega(d\ln{n})$~\cite{du2000combinatorial}, which is optimal in term of theoretical results. The decoding algorithm is usually included in the test design. In NAGT, Porat and Rothschild~\cite{porat2008explicit} first proposed explicit nonadaptive constructions using $O(d^2 \ln{n})$ tests with no efficient (sublinear to $n$) decoding algorithm. To have efficient decoding algorithm, says $\poly(d, \ln{n})$, while keeping the number of tests as low as possible, says $O(d^{1 + o(1)} \ln^{1 + o(1)}{n})$, several schemes have been proposed~\cite{ngo2011efficiently,cheraghchi2013noise,bui2018efficient}. Using probabilistic methods, Cai et al.~\cite{Cai2013:Grotesque} required only $O(d \ln{d} \cdot \ln{n})$ tests to find defective items in time $O(d(\ln{n} + \ln^2{d}))$.

From the genesis day of TGT, Damaschke~\cite{damaschke2006threshold} showed that the set of positive items can be identified with up to $g$ false positives and $g$ false negatives by using $\binom{n}{u}$ non-adaptive tests. Since the number of tests is quite large, Cheraghchi~\cite{cheraghchi2013improved} reduced it to $O(d^{g+2} \ln{d} \cdot \ln(n/d))$ tests. With the assumption that the number of defective items is exactly $d$, De Marco et al.~\cite{de2017subquadratic} and Chan et al.~\cite{chan2013stochastic} reduced the number of tests to $O(d^{3/2} \ln(n/d))$ and $O\left( \ln{\frac{1}{\epsilon}} \cdot d\sqrt{u} \ln{n}\right)$, respectively. D'yachkov et al.~\cite{d2013superimposed} could achieve $O \left( d^2 \ln{n} \cdot \frac{(u-1)! 4^u}{(u - 2)^u (\ln{2})^u}\right)$ tests as $n$ goes to infinity.

Although the authors in~\cite{cheraghchi2013improved,de2017subquadratic} and~\cite{d2013superimposed} proposed nearly optimal bounds on the number of tests, there are no decoding algorithms associated with their schemes. By setting that the number of defective items is exactly $d$ and $u = o(d)$, Chan et al.~\cite{chan2013stochastic} used $O\left( \ln{\frac{1}{\epsilon}} \cdot d\sqrt{u} \ln{n}\right)$ tests to identify defective items in time $O(n \ln{n} + n \ln{\frac{1}{\epsilon}})$, which is linear to the number of items, where $\epsilon \in (0, 1).$ Chen and Fu~\cite{chen2009nonadaptive} proposed schemes that can find the defective items using $O \left( \frac{d^d}{u^u (d-u)^{d-u}} \cdot d \ln{\frac{n}{d}} \right)$ tests in time $O(n^u \ln{n}) $, which is impractical as $n$ or $u$ grows. Bui et al.~\cite{bui2017efficiently} proposed a scheme that can identify the set of defective items with $t = O \left( \frac{d^d}{u^u (d-u)^{d-u}} \cdot d^3 \ln{n} \cdot \ln{\frac{n}{d}} \right)$ tests in sublinear time $t \times O(d^{11} \ln^{17}{n})$. However, the number of tests is quite large and the decoding time is not efficient for small $n.$ Recently, by setting $d = O(n^\beta)$ for $\beta \in (0, 1)$ and $u = o(d)$, Reisizadeh et al. \cite{reisizadehsub} use $\Theta(\sqrt{u} d \ln^3{n} )$ tests to identify all defective items in time $O(u^{1.5} d \ln^4{n} )$ w.h.p with the aid of a $O(u \ln{n}) \times \binom{n}{u}$ look-up matrix, which is unfeasible when $n$ or $u$ grows.

\subsection{Contribution}
\label{sub:intro:contri}

In this paper, we consider the case where $g = 0$, i.e., $\ell = u-1$ ($u \geq 2$). We inherit then improve the results in~\cite{bui2017efficiently}. The main idea in~\cite{bui2017efficiently} is to create two matrices: one for locating defective items, denoted an $h \times n$ matrix $\cG$, and the other one for identifying the defective items in each row of $\cG$, denoted an $(2k + 1) \times n$ matrix $\cA$. Let $\dec(\cA)$ be the decoding complexity of $\cA$ for each row of $\cG.$ After using a concatenation technique on $\cG$ and $\cA$, the final measurement matrix $\cT$ is used for identifying all defective items. The number of tests in $\cT$ is $h(2k+1)$ and the decoding complexity is $h \times \dec(\cA).$ Our contribution is to reduce $h$ and $\dec(\cA)$ while $k$ relatively remains same. As a result, the number of tests and the decoding complexity are significantly improved in accordance with existing results as shown in Table~\ref{tbl:comparison}.

Although Cheraghchi~\cite{cheraghchi2013improved}, De Marco et al.~\cite{de2017subquadratic}, and D'yachkov et al.~\cite{d2013superimposed} proposed nearly optimal bounds on the number of tests, there are no decoding algorithms associated with their schemes. Chen et al.~\cite{chen2009nonadaptive} required $\frac{d^{d + 1}}{u^u (d-u)^{d-u}} \ln{\frac{n}{d}}$ tests with the decoding complexity $O(n^u \ln{n})$, which is impractical. By setting that the number of defective items is exactly $d$ and $u = o(d)$, Chan et al.~\cite{chan2013stochastic} achieved a small number of tests $O\left( \ln{\frac{1}{\epsilon}} \cdot d\sqrt{u} \ln{n}\right).$ However, the decoding complexity is linear to the number of items, namely $O(n \ln{n} + n \ln{\frac{1}{\epsilon}})$, where $\epsilon > 0$ is the precision parameter. Recently, by setting $d = O(n^\beta)$ for $0 < \beta < 1$ and $u = o(d)$, Reisizadeh et al. \cite{reisizadehsub} can use $\Theta(\sqrt{u} d \ln^3{n} )$ tests to identify all defective items in time $O(u^{1.5} d \ln^4{n} )$. The main drawback of this approach (along with the conditions $d = O(n^\beta)$ and $u = o(d)$) is that a $O(u \ln{n}) \times \binom{n}{u}$ look-up matrix must be stored, which is unfeasible when $n$ and $u$ grow.

Our proposed scheme balances the trade-off between the number of tests and the decoding complexity. Moreover, there are no ``unnatural'' constrains on the number of defective items and the threshold. Specifically, the number of defective items is up to $d$ and $2 \leq u \leq d$. There are two approaches for balancing the number of tests and decoding complexity. First, the set of defective items can be identified with $O \left( h \times \frac{d^2\ln^2{n}}{\mathsf{W}^2(d\ln{n})} \right) \approx O\left( \frac{d^4 \ln^3{n}}{(\ln(d\ln{n}) - \ln(\ln(d\ln{n})))^2} \right)$ tests in time $O \left( \dec_0 \times h \right)$, where $h = O\left( \frac{d_0^2 \ln{\frac{n}{d_0}}}{(1-p)^2} \right)$, $\Wx = \Theta \left( \ln{x} - \ln{\ln{x}} \right)$, and $\dec_0 = O \left( \frac{d^{3.57} \ln^{6.26}{n}}{\mathsf{W}^{6.26}(d \ln{n})} \right) + O \left( \frac{d^6 \ln^4{n}}{\mathsf{W}^4(d \ln{n})} \right)$ for $d_0 = \max\{u, d - u \}$ and $p \in [0, 1).$ Second, the decoding time can be reduced to $O \left(\dec_1 \times h \right)$ if the number of tests is increased to $O\left( h \times \frac{d^2\ln^3{n}}{\mathsf{W}^2(d\ln{n})} \right) \approx O\left( \frac{d^4 \ln^4{n}}{(\ln(d\ln{n}) - \ln(\ln(d\ln{n})))^2} \right)$, where $\dec_1 = \max \left\{ \frac{d^2 \ln^3{n}}{\mathsf{W}^2(d\ln{n})}, \frac{ud \ln^4{n}}{\mathsf{W}^3(d\ln{n})} \right\}$. In addition, the proposed scheme is capable of handling $\Omega(ph/d_0)$ erroneous outcomes.

\begin{table*}[t]

\begin{center}
\scalebox{1}{
\begin{tabular}{|l|c|c|c|c|c|c|}
\hline
Scheme & \begin{tabular}{@{}c@{}} \#defective \\items \end{tabular} & \begin{tabular}{@{}c@{}} Threshold \\ $u$ \end{tabular} & \begin{tabular}{@{}c@{}} Error \\tolerance \end{tabular} & \begin{tabular}{@{}c@{}} Number of tests \\ $t$ \end{tabular}  & Decoding complexity & \begin{tabular}{@{}c@{}} Decoding \\ type \end{tabular} \\
\hline
\begin{tabular}{@{}c@{}} Cheraghchi~\cite{cheraghchi2013improved} \end{tabular} & $\leq d$ & $u \leq d$ & $\Omega(pt/d)$ & $O \left( \frac{d^2 \ln{\frac{n}{d}}}{(1-p)^2} \right)$ & $\times$ & $\times$ \\
\hline
\begin{tabular}{@{}c@{}} De Marco et al.~\cite{de2017subquadratic} \end{tabular} & $d$ & $u = O(\sqrt{d})$ & $\times$ &  $O \left( d^2 \cdot \sqrt{\frac{d-u}{du}} \cdot \ln{\frac{n}{d}} \right)$ & $\times$ & $\times$ \\
\hline
\begin{tabular}{@{}c@{}} D'yachkov et al.~\cite{d2013superimposed} \end{tabular} & $\leq d$ & $u \leq d$ & $\times$ & $O \left( d^2 \ln{n} \cdot \frac{(u-1)! 4^u}{(u - 2)^u (\ln{2})^u}\right)$ & $\times$ & $\times$ \\
\hline
\begin{tabular}{@{}c@{}} Chen et al.~\cite{chen2009nonadaptive} \end{tabular} & $\leq d$ & $u \leq d$ & $\times$ & $O \left( \alpha d \ln{\frac{n}{d}} \right)$ & $O(n^u \ln{n}) $ & Deterministic \\
\hline
\begin{tabular}{@{}c@{}} Bui et al.~\cite{bui2017efficiently} \end{tabular} & $\leq d$ & $2 \leq u \leq d$ & $\times$ & $O \left( \alpha d^3 \ln{n} \cdot \ln{\frac{n}{d}} \right)$ & $t \times O(d^{11} \ln^{17}n)$ & Deterministic \\
\hline
\begin{tabular}{@{}c@{}} Chan et al.~\cite{chan2013stochastic} \end{tabular} & $d$ & $u = o(d)$ & $\times$ & $O\left( \ln{\frac{1}{\epsilon}} \cdot d\sqrt{u} \ln{n}\right)$ & $O(n\ln{n} + n \ln{\frac{1}{\epsilon}})$ & Random \\
\hline
\begin{tabular}{@{}c@{}} Bui et al.~\cite{bui2017efficiently} \end{tabular} & $\leq d$ & $2 \leq u \leq d$ & $\times$ & $O \left( \alpha \left(u \ln{\frac{d}{u}} + \ln{\frac{1}{\epsilon}} \right) \cdot d^2 \ln{n} \right)$ & $t \times O(d^{11} \ln^{17}n)$ & Random \\
\hline
\begin{tabular}{@{}c@{}} Reisizadeh et al.~\cite{reisizadehsub} \end{tabular} & \begin{tabular}{@{}c@{}} $d = O(n^\beta)$ \\ for $0 < \beta < 1$ \end{tabular} & $u = o(d)$ & $\times$ & $\Theta(\sqrt{u} d \log_2^3{n} )$ & \begin{tabular}{@{}c@{}} $O(u^{1.5} d \log^4{n} )$ \\ with the aid of a \\$O(u \log_2{n}) \times \binom{n}{u}$ \\ look-up matrix \end{tabular} & Random \\
\hline
\begin{tabular}{@{}c@{}} \textbf{Proposed 1} \\ \textbf{(Corollary~\ref{cor:threshold:1})} \end{tabular} & $\leq d$ & $2 \leq u \leq d$ & $\Omega(ph/d_0)$ & $O \left( h \times \frac{d^2\ln^2{n}}{\mathsf{W}^2(d\ln{n})} \right)$ & $O \left( \dec_0 \times h \right)$ & Deterministic \\
\hline
\begin{tabular}{@{}c@{}} \textbf{Proposed 2} \\ \textbf{(Corollary~\ref{cor:threshold:2})} \end{tabular} & $\leq d$ & $2 \leq u \leq d$ & $\Omega(ph/d_0)$ & $O\left( h \times \frac{d^2\ln^3{n}}{\mathsf{W}^2(d\ln{n})} \right)$ & $O \left(\dec_1 \times h \right)$ & Deterministic \\
\hline
\end{tabular}}

\end{center}
\caption{Comparison with existing work. In this table, notation $\times$ means that the criterion does not hold for that scheme. For short notations, we set $\alpha = \frac{d^d}{u^u (d-u)^{d-u}}; \ \dec_0 = O \left( \frac{d^{3.57} \ln^{6.26}{n}}{\mathsf{W}^{6.26}(d \ln{n})} \right) + O \left( \frac{d^6 \ln^4{n}}{\mathsf{W}^4(d \ln{n})} \right),$ $h = O\left( \frac{d_0^2 \ln{\frac{n}{d_0}}}{(1-p)^2} \right)$, and $\dec_1 = \max \left\{ \frac{d^2 \ln^3{n}}{\mathsf{W}^2(d\ln{n})}, \frac{ud \ln^4{n}}{\mathsf{W}^3(d\ln{n})} \right\}$, where $d_0 = \max\{u, d - u \}$ and $p \in [0, 1).$ Function $\Wx$ is a Lambert W function, i.e., $\Wx\exp^{\Wx} = x$ for every $x > -1/\mathrm{e}$. Specifically, we have $\Wx = \Theta \left( \ln{x} - \ln{\ln{x}} \right)$.}

\label{tbl:comparison}
\end{table*}

\subsection{Organization}
\label{sub:org}

The paper is organized as follows. Section~\ref{sec:pre} presents some preliminaries on notations, problem definition, and $d$-disjunct matrices. Section~\ref{sec:review} reviews a previous work. Section~\ref{sec:proposed} presents how to improve the previous work and results. The final section summarizes the key points and addresses some open problems.

\section{Preliminaries}
\label{sec:pre}

For consistency, we use capital calligraphic letters for matrices, non-capital letters for scalars, bold letters for vectors, and capital letters for sets. All matrices are binary. Capital letters with asterisk is denoted for multisets in which elements may appear multiple times. For example, $D = \{1, 2 \} $ is a set and $S^* = \{1, 1, 2 \}$ is a multiset.

Let function $\add(\cdot)$ be a function that returns a multiset including all elements in the input sets/multisets. For example, suppose the input sets are $A^* = \{ 1, 1, 2\}$ and $B = \{2, 3, 4 \}$, then we have $\add(A^*, B) = \{1, 1, 2, 2, 3, 4 \}.$ Here are some of the notations used:

\begin{enumerate}
\item $n, d, \bX = (x_1, \ldots, x_n)^T$: number of items, maximum number of defective items, and binary representation of $n$ items.
\item $D = \{j_1, j_2, \ldots, j_{|D|} \}$: the set of defective items; cardinality of $D$ is $|D| \leq d$.
\item $\otimes, \odot$: operation related to $u$-NATGT and CNAGT, to be defined later.
\item $\cT$: $t \times n$ measurement matrix used to identify at most $d$ defective items in $u$-NATGT, where integer $t \geq 1$ is the number of tests.
\item $\cG = (g_{ij})$: $h \times n$ matrix, where $h \geq 1$.
\item $\cM = (m_{ij})$: a $k \times n$ $(d+1)$-disjunct matrix, where $k \geq 1.$
\item $\overline{\cM} = (\overline{m}_{ij})$: the $k \times n$ complementary matrix of $\cM$; $\overline{m}_{ij} = 1 - m_{ij}$.
\item $\cT_{i, *}, \cG_{i, *}, \cM_{i,*}, \cM_j$: row $i$ of matrix $\cT$, row $i$ of matrix $\cG$, row $i$ of matrix $\cM$, and column $j$ of matrix $\cM$, respectively.
\item $\cG \mid_{S}$: an $h \times |S|$ submatrix of an $h \times n$ matrix $\cG$ formed by restricting $\cG$ to the columns picked by $S$.
\item $\diag(\cG_{i, *}) = \diag(g_{i1}, \ldots, g_{iN})$: diagonal matrix constructed by input vector $\cG_{i, *}$.
\item $\supp(.)$: support index set of the input vector. For example, $\supp(1, 0, 1, 0, 0, 1) = \{1, 3, 6 \}$.
\item $\mathrm{e}, \ln$: base of natural logarithm, natural logarithm.
\item $|\cdot|$: weight; i.e, number of non-zero entries of input vector or cardinality of input set.
\item $\lceil x \rceil, \lfloor x \rfloor$: ceiling and floor functions of $x$.
\end{enumerate}

\subsection{Problem definition}
\label{sub:probDef}

In a population of $n$ items, up to $d$ items, which are called \textit{defective items}, satisfy some certain properties. In $u$-NATGT, a subset containing at least $u$ defective items satisfies the certain properties while a subset containing up to $u - 1$ items does not hold. It is equivalent to the fact that the outcome of a test (for checking the certain properties) on a subset of $n$ items is positive if the subset has at least $u$ defective items, and negative otherwise. Our goal is to identify the set of defective items with as few tests as possible and as quick as possible.

Let $D$, $[n] = \{1, 2, \ldots, n \}$, and $S$ be the set of defective items, the index set of $n$ items, and an arbitrary subset of $[n]$, respectively. Formally, the outcome of a test (a test in short) on the subset $S$ is positive if $|D \cap S| \geq u$ and negative if $|D \cap S| < u$. For any $t$ non-adaptive tests, they can be represented as a $t \times n$ binary measurement matrix $\cT =(t_{ij})$, where column $j$ and row $i$ represent for item $j$ and test $i$, respectively. An entry $t_{ij}=1$ means that item $j$ belongs to test $i$, and $t_{ij}=0$ means otherwise. Binary vector $\bX = (x_1,\ldots,x_n)^T$ represents for $n$ items, where $x_j=1$ indicates that item $j$ is defective, and $x_j=0$ indicates otherwise. From the assumption, it is clear that $D = \supp(\bX)$ and $|D| \leq d$. The outcome of the $t$ tests is $\bY=(y_1, \ldots, y_t)^T$, where $y_i=1$ if test $i$ is positive and $y_i=0$ otherwise. The procedure to get the outcome vector $\bY$ is called the \textit{encoding procedure}. The procedure used to identify defective items from $\bY$ is called the \textit{decoding procedure.} The relationship between $\bX, \cT,$ and $\bY$ can be represented as follows:

\begin{equation}
\label{eqn:thresholdGT}
\bY = \cT \otimes \bX = \begin{bmatrix}
\cT_{1, *} \otimes \bX \\
\vdots \\
\cT_{t, *} \otimes \bX
\end{bmatrix} = \begin{bmatrix}
y_1 \\
\vdots \\
y_t
\end{bmatrix}
\end{equation}
where $\otimes$ is a notation for the test operation in $u$-NATGT; namely, $y_i = \cT_{i, *} \otimes \bX = 1$ if $|\supp(\bX) \cap \supp(\cT_{i, *})| = |D \cap \supp(\cT_{i, *})| \geq u$, and $y_i = 0$ otherwise for $i=1, \ldots, t$.

\subsection{Disjunct matrices}
\label{sub:disjunct}
When $u = 1$, $u$-NATGT reduces to CNAGT. To distinguish CNAGT and $u$-NATGT, we change notation $\otimes$ to $\odot$ and use a $k \times n$ measurement matrix $\cM$ instead of the $t \times n$ matrix $\cT$. The outcome vector $\bY$ is equal to
\begin{equation}
\label{eqn:disjunct}
\bY = \cM \odot \bX = \begin{bmatrix}
\cM_{1, *} \odot \bX \\
\vdots \\
\cM_{k, *} \odot \bX
\end{bmatrix}
= \begin{bmatrix}
\bigvee_{j=1}^{n} x_j \wedge m_{1j} \\
\vdots \\
\bigvee_{j=1}^{n} x_j \wedge m_{kj}
\end{bmatrix} = \begin{bmatrix}
y_1 \\
\vdots \\
y_k
\end{bmatrix} \notag
\end{equation}
where $\odot$ is the Boolean operator for vector multiplication in which multiplication is replaced with the AND ($\wedge$) operator and addition is replaced with the OR ($\vee$) operator, and $y_i = \cM_{i, *} \odot \bX = \bigvee_{j=1}^{n} x_j \wedge m_{ij}$ for $i = 1, \ldots, k$. Similarly to~\eqref{eqn:thresholdGT}, $y_i = 1$ if if $|\supp(\bX) \cap \supp(\cM_{i, *})| = |D \cap \supp(\cM_{i, *})| \geq 1$, and $y_i = 0$ when $|\supp(\bX) \cap \supp(\cM_{i, *})| = 0$.

The formal definition of a $d$-disjunct matrix is as follows.

\begin{definition}
Matrix $\cM$ is $d$-disjunct if for any two disjoint subsets $S_1, S_2 \subset [n]$ such that $|S_1| = d$ and $|S_2| = 1$, there exists at least $1$ row in which the column in $S_2$ has 1 while all the columns in $S_1$ have 0's, i.e., $\left\vert \bigcap_{j \in S_2} \supp \left( \cM_j \right) \big\backslash \bigcup_{j \in S_1} \supp \left( \cM_j \right) \right\vert \geq 1$.
\label{def:disjunct}
\end{definition}

When $\cM$ is $d$-disjunct, vector $\bX$ can always be recovered from $\bY = \cM \odot \bX$. Because of later use, we only pay attention for any $d$-disjunct matrix whose columns can be computed in time $\poly(k)$.

By numerical results, Bui et al.~\cite{bui2018efficient} showed that the number of tests in nonrandom construction (each column can be deterministically generated without using probability) is the best for practice (albeit it is not good in term of complexity). Therefore, we prefer to use that result here.

\begin{theorem}~\cite[Theorem 8]{bui2018efficient}
\label{thr:mainNonrandom}
Let $1 \leq d \leq n$ be integers. Then there exists a nonrandom $d$-disjunct matrix $\cM$ with $k = O \left(\frac{d^2 \ln^2{n}}{\mathsf{W}^2(d\ln{n})} \right) = O \left(\frac{d^2 \ln^2{n}}{(\ln(d\ln{N}) - \ln{\ln(d\ln{n})})^2} \right)$. Each column of $\cM$ can be computed in time $O(k^{1.5}/d^2)$. Then matrix $\cM$ can be used to identify up to $d^\prime$ defective items, where $d^\prime \geq \left\lfloor \frac{d}{2} \right\rfloor + 1$, in time $\dec_0 = O \left( \frac{d^{3.57} \ln^{6.26}{n}}{\mathsf{W}^{6.26}(d \ln{n})} \right) + O \left( \frac{d^6 \ln^4{n}}{\mathsf{W}^4(d \ln{n})} \right).$ When $d$ is the power of 2, $d^\prime = d - 1$.
\end{theorem}

The decoding complexity can be reduce by increasing the number of tests as follows:

\begin{theorem}~\cite[Corollary 3]{bui2018efficient}
\label{thr:mainNonrandom2}
Let $1 \leq d \leq n$ be integers. There exists a nonrandom $k \times n$ measurement matrix $\cT$ with $k = O \left( \frac{d^2 \ln^3{n}}{\mathsf{W}^2(d \ln{n})} \right) = O \left(\frac{d^2 \ln^3{n}}{(\ln(d\ln{n}) - \ln{\ln(d\ln{n})})^2} \right)$, which is used to identify at most $d$ defective items in time $O(k)$. Moreover, each column of $\cT$ can be computed in time $O \left(\frac{d \ln^4{n}}{\mathsf{W}^3(d\ln{n})} \right)$.
\end{theorem}

We denote the procedure of getting $\bX$ from $\cM \odot \bX$ as $\bX = \decode(\cM, \cM \odot \bX)$.

\section{Review of Bui et al.'s scheme}
\label{sec:review}

The scheme proposed by Bui et al.~\cite{bui2017efficiently} is reviewed here. The authors created two inseparable matrices: an indicating matrix $\cG$ and a defective-solving matrix $\cA$. The task of $\cA$ is to recover any $\bX$ from $\bY = \cA \otimes \bX$ if $|\bX|$ is equal to $u$. The task of $\cG$ is to ensure that there exists $\kappa$ rows, e.g., $i_1, i_2, \ldots, i_\kappa$ such that $|D \cap \supp(\cG_{i_1})| = \cdots = |D \cap \supp(\cG_{i_\kappa})| = u$ and $(D \cap \supp(\cG_{i_1})) \cup \ldots \cup (D \cap \supp(\cG_{i_\kappa})) = D.$ The final measurement matrix $\cT$ generated from $\cG$ and $\cA$ then is used to identify all defective items. The details of this scheme is described as the following.

\subsection{When the number of defective items equals the threshold}
\label{sub:specialCase}

The authors first considered a special case in which the number of defective items equals the threshold, i.e., $|\bX| = u$. Let $\cM = (m_{ij})$ be a $k \times n$ $(d+1)$-disjunct matrix as described in Section~\ref{sub:disjunct}. Then a measurement matrix is created as
\begin{equation}
\label{eqn:elementaryMatrix}
\cA = \begin{bmatrix}
\cM \\
\overline{\cM}
\end{bmatrix}
\end{equation}
where $\overline{\cM} = (\overline{m}_{ij})$ is the complement matrix of $\cM$, $\overline{m}_{ij} = 1 - m_{ij}$ for $i = 1, \ldots, k$ and $j = 1, \ldots, n$.

Given measurement matrix $\cA$ and a representation vector of $u$ defective items $\bX$ ($|\bX| = u$), what we observe is $\bZ = \cA \otimes \bX$. The objective is to recover $\bY^\prime = \cM \odot \bX = (y^\prime_1, \ldots, y^\prime_k)^T$ from $\bZ$. Then $\bX$ can be recovered by using Theorem~\ref{thr:mainNonrandom} or~\ref{thr:mainNonrandom2}.

Assume that the outcome vector is
\begin{equation}
\label{eqn:special1}
\bZ = \cA \otimes \bX = \begin{bmatrix}
\cM \otimes \bX \\
\overline{\cM} \otimes \bX
\end{bmatrix} = \begin{bmatrix}
\bY \\
\overline{\bY}
\end{bmatrix}
\end{equation}
where $\bY = \cM \otimes \bX = (y_1, \ldots, y_k)^T$ and $\overline{\bY} = \overline{\cM} \otimes \bX = (\overline{y}_1, \ldots, \overline{y}_k)^T$. Then vector $\bY^\prime = \cM \odot \bX$ is always obtained from $\bZ$ by using the following rules:

\begin{enumerate}
\item If $y_l = 1$, then $y^\prime_l = 1$.
\item If $y_l = 0$ and $\overline{y}_l = 1$, then $y^\prime_l = 0$.
\item If $y_l = 0$ and $\overline{y}_l = 0$, then $y^\prime_l = 1$.
\end{enumerate}

Therefore, vector $\bX$ can always be recovered.

\subsection{Encoding procedure}
\label{sub:enc}

After preparing matrix $\cA$ for identifying exactly $u$ defective items, the next task is to create matrix $\cG$ and the final measurement matrix $\cT$. The authors generated matrix $\cG$ such that there exists $\kappa = \left\lceil \frac{|D|}{u} \right\rceil$ rows, denoted as $i_1, i_2, \ldots, i_\kappa$, satisfying (i) $|D \cap \supp(\cG_{i_1})| = \cdots = |D \cap \supp(\cG_{i_\kappa})| = u$ and (ii) $(D \cap \supp(\cG_{i_1})) \cup \ldots \cup (D \cap \supp(S_{i_\kappa})) = D.$ Then the final measurement matrix $\cT$ of size $(2k + 1)h \times n$ is created as follows:

\begin{equation}
\label{eqn:meausrementMatrix}
\cT = \begin{bmatrix}
\cG_{1, *} \\
\cA \times \diag(\cG_{1, *}) \\
\vdots \\
\cG_{h, *} \\
\cA \times \diag(\cG_{h, *})
\end{bmatrix}
= \begin{bmatrix}
\cG_{1, *} \\
\cM \times \diag(\cG_{1, *}) \\
\overline{\cM} \times \diag(\cG_{1, *}) \\
\vdots \\
\cG_{h, *} \\
\cM \times \diag(\cG_{h, *}) \\
\overline{\cM} \times \diag(\cG_{h, *})
\end{bmatrix}
\end{equation}

The vector observed using $u$-NATGT after performing the tests given by the measurement matrix $\cT$ is
\begin{eqnarray}
\bY = \cT \otimes \bX &=& \begin{bmatrix}
\cG_{1, *} \\
\cA \times \diag(\cG_{1, *}) \\
\vdots \\
\cG_{h, *} \\
\cA \times \diag(\cG_{h, *})
\end{bmatrix} \otimes \bX
= \begin{bmatrix}
\cG_{1, *} \otimes \bX \\
\cA \otimes \bX_1 \\
\vdots \\
\cG_{h, *} \otimes \bX \\
\cA \otimes \bX_h
\end{bmatrix} \notag \\
&=& \begin{bmatrix}
\cG_{1, *} \otimes \bX\\
\cM \otimes \bX_1  \\
\overline{\cM} \otimes \bX_1 \\
\vdots \\
\cG_{h, *} \otimes \bX\\
\cM \otimes \bX_h \\
\overline{\cM} \otimes \bX_h \\
\end{bmatrix}
= \begin{bmatrix}
y_1 \\
\bY_1 \\
\overline{\bY}_1 \\
\vdots \\
y_h \\
\bY_h \\
\overline{\bY}_h
\end{bmatrix}
= \begin{bmatrix}
y_1 \\
\bZ_1 \\
\vdots \\
y_h \\
\bZ_h
\end{bmatrix} \label{eqn:encoding}
\end{eqnarray}
where $\bX_i = \diag(\cG_{i, *}) \times \bX$, $y_i = \cG_{i, *} \otimes \bX$, $\bY_i = \cM \otimes \bX_i = (y_{i1}, \ldots, y_{ik})^T$, $\overline{\bY}_i = \overline{\cM} \otimes \bX_i = (\overline{y}_{i1}, \ldots, \overline{y}_{ik})^T$, and $\bZ_i = [\bY_i^T \ \overline{\bY}_i^T]^T$ for $i = 1, 2, \ldots, h$.

Vector $\bX_i$ is the vector representing the defective items in row $\cG_{i, *}$. Therefore, we have $|\supp(\bX_i)| \leq d$, and $y_i = 1$ if and only if $|\supp(\bX_i)| \geq u$.

\subsection{The decoding procedure}
\label{sub:dec}

The decoding procedure follows the properties of $\cM$, $\cG$, and $\cT$. From~\eqref{eqn:encoding}, the authors presumed the cardinality of every $\bX_i$ is $u$ ($|\supp(\bX_i)| = |D \cap \supp(\cG_{i, *})| = u$). Then by using the scheme in section~\ref{sub:specialCase} for each $\bZ_i$, they could recover $\bY_i^\prime = (y_{i1}^\prime, \ldots, y_{ik}^\prime)^T$, which is presumed to be $\cM \odot \bX_i$. A vector is obtained by using $\decode(\bY_i^\prime, \cM)$. However, because $|\supp(\bX_i)|$ may not equal $u$, the vector obtained from $\decode(\bY_i^\prime, \cM)$ may not be $\bX_i$. They thus used a sanitary procedure to eliminate this case based on the properties of $\cM.$

The whole decoding algorithm is summarized as Algorithm~\ref{alg:decodingThreshold}. It is briefly explained as follows: Line~\ref{alg:decodingThreshold:enum} enumerates $h$ rows of $\cG.$ Line~\ref{alg:decodingThreshold:checkU} checks if there are at least $u$ defective items in row $\cG_{i, *}$. Lines~\ref{alg:decodingThreshold:convert2CNAGT_start} to~\ref{alg:decodingThreshold:convert2CNAGT_end} calculate $\bY_i^\prime$, and line~\ref{alg:decodingThreshold:possibleDefectives} gets a possible set of defective items. Lines~\ref{alg:decodingThreshold:sanitary_start} to~\ref{alg:decodingThreshold:sanitary_end} check whether all items in $G_i$ are truly defective then adds them into the defective set $D.$ Finally, line~\ref{alg:decodingThreshold:finish} returns the defective set.

\begin{algorithm}
\caption{$\mathrm{FindDefectiveItems}(\bY, \cM)$: Decoding procedure for $u$-NATGT with no error-tolerance.}
\label{alg:decodingThreshold}
\textbf{Input:} Outcome vector $\bY$, $\cM$.\\
\textbf{Output:} The set of defective items $D.$

\begin{algorithmic}[1]
\State $D = \emptyset$.
\For {$i=1$ to $h$} \label{alg:decodingThreshold:enum}
	\If {$y_i = 1$} \label{alg:decodingThreshold:checkU}
		\For {$l = 1$ to $k$} \label{alg:decodingThreshold:convert2CNAGT_start}
			\State \textbf{If} {$y_{il} = 1$} \textbf{then} $y^\prime_{il} = 1$ \textbf{end if}
			\State \textbf{If} {$y_{il} = 0$ and $\overline{y}_{il} = 1$} \textbf{then} $y^\prime_{il} = 0$ \textbf{end if}
			\State \textbf{If} {$y_{il} = 0$ and $\overline{y}_{il} = 0$} \textbf{then} $y^\prime_{il} = 1$ \textbf{end if}
		\EndFor \label{alg:decodingThreshold:convert2CNAGT_end}
		\State $G_i = \supp(\decode(\cM, \bY^\prime_i )).$ \label{alg:decodingThreshold:possibleDefectives}
		\If {$|G_i| = u$ and $\bigvee_{j \in G_i} \cM_j \equiv \bY_i$} \label{alg:decodingThreshold:sanitary_start}
			\State $D = D \cup G_i$. \label{alg:decodingThreshold:add_items}
		\EndIf \label{alg:decodingThreshold:sanitary_end}
	\EndIf	
\EndFor
\State Return $D$. \label{alg:decodingThreshold:finish}
\end{algorithmic}
\end{algorithm}

The decoding complexity of this algorithm is described as follows.

\begin{theorem}~\cite[Simplified version of Theorem 3]{bui2017efficiently}
\label{thr:general}
Let $2 \leq u \leq d < n$ be integers and $D$ be the defective set. Suppose that an $h \times n$ matrix $\cG$ contains $\kappa$ rows, denoted as $i_1, \ldots, i_\kappa$, such that (i) $|D \cap \supp(\cG_{i_1})| = \cdots = |D \cap \supp(\cG_{i_\kappa})| = u$ and $(D \cap \supp(\cG_{i_1})) \cup \ldots \cup (D \cap \supp(\cG_{i_\kappa})) = D$. Suppose that a $k \times n$ matrix $\cM$ is an $(d + 1)$-disjunct matrix that can be decoded in time $O(\sfA)$ and each column of $\cM$ can be generated in time $O(\sfB)$. Then an $(2k + 1)h \times n$ measurement matrix $\cT$, as defined in \eqref{eqn:meausrementMatrix}, can be used to identify up to $d$ defective items in $u$-NATGT in time $O(h \times (\sfA + u \sfB))$.
\end{theorem}

Algorithm~\ref{alg:decodingThreshold} is denoted as $\mathrm{FindDefectiveItems}^*(\bY, \cM)$ if we substitute $D$ by multiset $D^*$ and replace Line~\ref{alg:decodingThreshold:add_items} as ``$D^* = \add(D^*, G_i).$''; i.e., the output of $\mathrm{FindDefectiveItems}^*(\bY, \cM)$ may have duplicated items which are used to handle the presence of erroneous outcomes.

\section{Proposed scheme}
\label{sec:proposed}

We improve the result in Theorem~\ref{thr:general} by extending it to handle erroneous outcomes. Then, for its instantiations, we minimize $h$ and $\decode(\cM, \bY)$ while $k$ relatively remains same, where $\bY$ is some input vector. As a result, the number of tests and the decoding complexity are significantly improved.

To achieve this goal, we define a good measurement matrix as follows:

\begin{definition}
Let $2 \leq u \leq d < n$ be integers and $D$ be the defective set, where $|D| \leq d$. An $(n, d, u; e)$-measurement matrix $\cG$ is good if there exists $\varphi$ rows, e.g., $i_1, \ldots, i_\varphi$, such that:
\begin{enumerate}
\item $|D_{i_1}| = |D_{i_2}| = \ldots = |D_{i_\varphi}| = u$, where $D_{i_l} = \supp(\cG_{i_l}) \cap D$ for $l = 1, \ldots, \varphi$.
\item $D = D_{i_1} \cup \ldots \cup D_{i_\varphi}$.
\item Any item in $D$ appears more than $e$ times in $\add(D_{i_1}, \ldots, D_{i_\varphi}).$
\end{enumerate}
\label{def:goodMeasurement}
\end{definition}

The matrix $\cG$ in Theorem~\ref{thr:general} is a good $(n, d, u; 0)$ measurement matrix, i.e. erroneous outcomes are not considered. Intuitively, a good $(n, d, u; 2e)$-measurement matrix can handle up to $e$ erroneous outcome. We then show how to efficiently construct a good measurement matrix in the next section.

\subsection{On construction of a good measurement matrix}
\label{sub:improveDecThreshold:construction}

We first state the notation of threshold disjunct matrices proposed by Cheraghchi~\cite{cheraghchi2013improved}.

\begin{definition}~\cite[Definition 6]{cheraghchi2013improved}
\label{def:threshDisjunct}
A Boolean matrix $\cG$ with $n$ columns is called $(d, u; e)$-regular if for every subset of columns $S \subseteq [n]$ (called the critical set) and every $Z \subseteq [n]$ (called the zero set) such that $u \leq |S| \leq d$, $|Z| \leq |S|$, $S \cap Z = \emptyset$, there are more than $e$ rows of $\cG$ at which $\cG \mid_{S}$ has weight exactly $u$ and (at the same rows) $\cG \mid_{Z}$ has weight zero. Any such row is said to $u$-satisfy $S$ and $Z$. If, in addition, for every distinguished column $j \in S$, more than $e$ rows of $\cG$ both $u$-satisfy $S$ and $Z$ and have a 1 at the $j$th column, the matrix is called threshold $(d, u; e)$-disjunct (and the corresponding ``good'' rows are said to $u$-satisfy $j, S,$ and $Z$).
\end{definition}

The following lemma shows that a threshold $(\max\{u, d - u\}, u; e)$-disjunct is a good $(n, d, u; e)$-measurement matrix.
\begin{lemma}
\label{lem:threshDisjunct_GoodMatrix}
Let $0 < u \leq d \leq n$, and $0 \leq e$ be integers. Then any threshold $(\max\{u, d - u\}, u; e)$-disjunct matrix is a good $(n, d, u; e)$-measurement matrix.
\end{lemma}

\begin{proof}
Let $D = \{d_1, \ldots, d_{\kappa}\}$ be the defective set, where $|D| = \kappa \leq d$. Suppose $\cG$ is any threshold $(\max\{u, d - u\}, u; e)$-disjunct. We break down the condition $u \leq d$ into two categories: $u \leq d \leq 2u$ and $d \geq 2u + 1$.

When $u \leq d \leq 2u$, we have $\max\{u, d - u \} = u$. Let $D_1$ and $D_2$ be the two subsets of $D$ such that $D = D_1 \cup D_2$ and $|D_1| = |D_2| = u$. Since $\cG$ is threshold $(u, u; e)$-disjunct and $d \leq 2u$, we have $|D \setminus D_1| \leq |D_2| = u$ and $|D \setminus D_2| \leq |D_1| = u$. Therefore, there are more than $e$ rows of $\cG$ at which $\cG \mid_{D_1}$ ($\cG \mid_{D_2}$) has weight exactly $u$ and (at the same rows) $\cG \mid_{D \setminus D_1}$ ($\cG\mid_{D \setminus D_2}$) has weight zero. Suppose that the number of these rows in $\cG \mid_{D_1}$ and $\cG\mid_{D_2}$ is $\varphi.$ Since $D = D_1 \cup D_2$, each item in $D$ appears more than $e$ times in these $\varphi$ rows. Matrix $\cG$ is thus a good $(n, d, u; e)$-measurement matrix.

When $d \geq 2u + 1$, we have $\max\{u, d-u \} = d-u$. Let $S \subset [n]$ and $|S| = d - u$. Since $2u + 1 \leq d$, for every $Z \subset [n], |Z| = u + 1 \leq |S| = d - u, S \cap Z = \emptyset$, we have $|Z| + |S| = d + 1 > d$. We choose a collection of sets of defective items as follows: $P_l = \{j_l \}$ for $l = 1, \ldots, \kappa$.

We then choose $\kappa$ subsets $S_l$ and $Z_l$ for $l = 1, \ldots, \kappa$ to show that $D = \cup_{l = 1}^\kappa S_l$ and there are more than $e$ rows of $\cG$ at which $\cG \mid_{S_l}$ has weight exactly $u$ and (at the same rows) $\cG \mid_{Z_l}$ has weight zero. The first condition ensures that all defective items are included in the selected subsets. The second condition is equivalent to the statement that there are more than $e$ rows containing exactly $u$ defective items.

To prove that, two cases are needed to be considered here: $u \leq |D| < d-u$ and $d - u \leq |D| \leq d$. For the former case, choose $S_l \cap [n] = D$ and $|S_l| = d - u$, and $Z_l \subset [n] \setminus S_l$ and $|Z_l| = u + 1$. For the latter case, we set 

\begin{itemize}
\item $S_l = P_l \cup D_l$, where $D_l \subseteq D\setminus P_l, P_l \cap D_l = \emptyset$, and $|D_l| = d - u - 1$.
\item $Z_l \subset [n] \setminus S_l$, $D \setminus S_l \subseteq Z_l$, and $|Z_l| = u + 1$.
\end{itemize}

It is obvious that $|Z_l| = u + 1 \leq |S_l| = d - u$ and $S_l \cap Z_l = \emptyset.$ Moreover, since $u \leq |S_l| \leq d$, we have that $S_l$ is a critical set and $Z_l$ is a zero set in $\cG$. Thus, there are more than $e$ rows of $\cG$ at which $\cG \mid_{S_l}$ has weight exactly $u$, $\cG \mid_{Z_l}$ has weight zero at the same rows, and $\cG \mid_{P_l}$ has weight one at the same rows. Let denote $e + 1$ rows of these rows as $r_{l_1}, \ldots, r_{l_e}, r_{l_{e+1}}$.

Finally, we have:
\begin{enumerate}
\item $|D_{r_{1_1}}| = \ldots = |D_{r_{1_{e+1}}}| = \ldots = |D_{r_{\kappa_1}}| = \ldots = |D_{r_{\kappa_{e+1}}}| = u$, where $D_{r_{1_x}} = \supp(\cG_{r_{1_x}}) \cap S_l = \supp(\cG_{r_{1_x}}) \cap D$ for $l = 1, \ldots, \kappa$ and $x = 1, \ldots, e+1$.
\item $D = D_{r_{1_1}} \cup \ldots \cup D_{r_{1_{e+1}}} \cup \ldots \cup D_{r_{\kappa_1}} \cup \ldots \cup D_{r_{\kappa_{e+1}}}$.
\item Any item in $D$ appears more than $e$ times in $\add(D^*, D_{r_{1_1}}, \ldots, D_{r_{1_{e+1}}}, \ldots, D_{r_{\kappa_1}}, \ldots, D_{r_{\kappa_{e+1}}}).$
\end{enumerate}

According to Definition~\ref{def:goodMeasurement}, matrix $\cG$ is a good $(n, d, u; e)$-measurement matrix.
\end{proof}

Cheraghchi~\cite{cheraghchi2013improved} proposed a good construction on a threshold disjunct matrix as follows.

\begin{lemma}~\cite[Lemma 23]{cheraghchi2013improved}
\label{lem:dgu}
For every $p \in [0, 1)$ and integer parameter $0 < u \leq d < n$, there exists an $h \times n$ threshold $(d, u; \Omega(ph/d))$-disjunct matrix with probability $1 - o(1)$, where $h = O(d^2 \left( \ln{\frac{n}{d}} \right)/(1 - p)^2)$.
\end{lemma}

Because of Lemma~\ref{lem:threshDisjunct_GoodMatrix} and~\ref{lem:dgu}, we get the following theorem:

\begin{theorem}
\label{thr:goodMatrix}
For every $p \in [0, 1)$ and integer parameter $0 < u \leq d < n,$ there exists an $h \times n$ good $(n, d, u; \Omega(ph/d_0))$-measurement matrix with probability $1 - o(1)$, where $h = O(d_0^2 \left( \ln{\frac{n}{d_0}} \right)/(1 - p)^2)$ and $d_0 = \max\{u, d- u \}$.
\end{theorem}

\subsection{Encoding procedure}
\label{sub:improvedDecThreshold:enc}

We will get the measurement matrix with low number of tests that can tackle up to $e$ erroneous outcomes as follows. Suppose that $\cG$ is either an $k \times n$ $(d+1)$-disjunct matrix in Theorem~\ref{thr:mainNonrandom} or Theorem~\ref{thr:mainNonrandom2}, and $\cG$ is an $h \times n$ good $(n, d, u; 2e = \Omega(ph/d_0))$-measurement matrix in Theorem~\ref{thr:goodMatrix}. Then the final measurement matrix $\cT$ is generated as in~\eqref{eqn:meausrementMatrix}. Note that $h = O(d_0^2 \left( \ln{\frac{n}{d_0}} \right)/(1 - p)^2)$ for some $p \in [0, 1)$ and $d_0 = \max\{u, d- u \}$.

\subsection{The decoding procedure}
\label{sub:improvedDecThreshold:decoding}

The decoding procedure is summarized as Algorithm~\ref{alg:threshold}. The procedure is as similar to the procedure in Algorithm~\ref{alg:decodingThreshold}. However, the input matrices $\cM$ and $\cG$ are different from the ones in Algorithm~\ref{alg:threshold}. Step~\ref{alg:threshold:init} initializes the defective set as an empty set. Then Step~\ref{alg:threshold:getPotentialDefecs} adds all potential defectives to set $R^*$ by recalling Algorithm~\ref{alg:decodingThreshold}. Step~\ref{alg:threshold:removeDuplicates} scans all elements in $R^*$ at which an item is declared as a defective item if it appears at least $e + 1$ times in Steps~\ref{alg:threshold:checkFrequency} to ~\ref{alg:threshold:endAddingDefective}. Step~\ref{alg:threshold:theEnd} simply returns the defective set.

\begin{algorithm}[ht]
\caption{$\mathrm{DecNATGT}(\bY, \cM, \overline{\cM}, e)$: Decoding procedure for $u$-NATGT with error-tolerance.}
\label{alg:threshold}
\textbf{Input:} Outcome vector $\bY$, $\cM$.\\
\textbf{Output:} The set of defective items $D$.

\begin{algorithmic}[1]
\State $D = \emptyset$. \Comment{Initialize defective set.} \label{alg:threshold:init}
\State $R^* = \mathrm{FindDefectiveItems}^*(\bY, \cM)$. \label{alg:threshold:getPotentialDefecs} \Comment{Get all potential defectives.}
\Foreach {$x \in R^*$} \label{alg:threshold:removeDuplicates} \Comment{Remove false positives.}
	\If {$x$ appears in $R^*$ at least $e + 1$ times} \label{alg:threshold:checkFrequency}
		\State $D = D \cup \{ x \}$. \Comment{$x$ is the true defective item.}
		\State Remove all elements that equal $x$ in $R^*$.
	\EndIf \label{alg:threshold:endAddingDefective}
\EndForeach \label{alg:threshold:removeDuplicates2}
\State \Return $D$. \Comment{Return set of defective items.} \label{alg:threshold:theEnd}
\end{algorithmic}
\end{algorithm}

\subsection{Correctness of the decoding procedure}
\label{sub:improvedDecThreshold:correctness}

Let consider Step~\ref{alg:threshold:getPotentialDefecs}. Because $\cG$ is a good $(n, d, u; 2e)$-measurement matrix, there exists $\varphi$ rows, e.g., $i_1, \ldots, i_\varphi$, such that:

\begin{enumerate}
\item $|D_{i_1}| = |D_{i_2}| = \ldots = |D_{i_\varphi}| = u$, where $D_{i_l} = \supp(\cG_{i_l}) \cap D$ for $l = 1, \ldots, \varphi$.
\item $D = D_{i_1} \cup \ldots \cup D_{i_\varphi}$.
\item Any item in $D$ appears more than $2e$ times in $\add(D_{i_1}, \ldots, D_{i_\varphi}).$
\end{enumerate}

Therefore, any defective item will appear at least $2e + 1$ times in $R^*$ if there is no error in test outcomes.

If there are up to $e$ errors in the outcome vector $\bY$, any false defective cannot appear more than $e$ times in $R^*$. Therefore, a defective item will appear at least $2e + 1 - e = e + 1$ times in $R^*.$ Steps~\ref{alg:threshold:removeDuplicates} to~\ref{alg:threshold:removeDuplicates2} remove all false or duplicated items in $R^*.$ Finally, Step~\ref{alg:threshold:theEnd} simply returns the defective set.

\subsection{Decoding complexity}
\label{sub:improvedDecThreshold:complexity}

The maximum cardinality of $R^*$ in Step~\ref{alg:threshold:getPotentialDefecs} is $uh.$ Therefore, Steps~\ref{alg:threshold:removeDuplicates} to~\ref{alg:threshold:removeDuplicates2} takes $O(uh)$ time. Then Theorem~\ref{thr:general} can be revised to tackle erroneous outcomes as follows.

\begin{theorem}
\label{thr:generalThreshold2}
Let $2 \leq u \leq d \leq n$ be integers. Suppose that matrix $\cG$ is an $h \times n$ good $(n, d, u; 2e)$-measurement matrix and matrix $\cM$ is a $k \times n$ $(d + 1)$-disjunct matrix that can be decoded in time $O(A)$. Each column of $\cM$ can be generated in time $O(\sfB)$. Then an $(2k + 1)h \times n$ measurement matrix $\cT$, as defined in \eqref{eqn:meausrementMatrix}, can be used to identify up to $d$ defective items in $u$-NATGT in time $O(h \times (\sfA + u \sfB)) + O(uh) = O(h \times (\sfA + u \sfB))$ in the presence of up to $e$ erroneous outcomes.
\end{theorem}

\subsection{Instantiations of decoding complexity}
\label{sub:improvedDecThreshold:instan}

We instantiate Theorem~\ref{thr:generalThreshold2} by choosing $\cG$ as an $h \times n$ good $(n, d, u; \Omega(ph/d_0))$-measurement matrix in Theorem~\ref{thr:goodMatrix}, i.e., $h = O(d_0^2 \left( \ln{\frac{n}{d_0}} \right)/(1 - p)^2)$ where $d_0 = \max \{u, d-u \}$ and $p \in [0, 1).$ If $\cM$ is an $(d+1)$-disjunct matrix in Theorem~\ref{thr:mainNonrandom}, we have:

\begin{corollary}
\label{cor:threshold:1}
Let $2 \leq u \leq d < n$ be integers, $d_0 = \max \{u, d-u \}$, and some $p \in [0, 1)$. There exits a $t \times n$ measurement matrix $\cT$ such that up to $d$ defective items in $u$-NATGT can be identified in time $O \left( \dec_0 \times h \right)$ in the presence of up to $e = \Omega(ph/d_0)$ erroneous outcomes, where $t = O\left( h \times \frac{d^2\ln^2{n}}{\mathsf{W}^2(d\ln{n})} \right)$ and $h = O(d_0^2 \left( \ln{\frac{n}{d_0}} \right)/(1 - p)^2).$
\end{corollary}

\begin{proof}

Since matrix $\cM$ is a $k \times n$ $(d+1)$-disjunct matrix in Theorem~\ref{thr:mainNonrandom}, we have:
\begin{itemize}
\item $k = O \left(\frac{d^2 \ln^2{n}}{\mathsf{W}^2(d\ln{n})} \right)$.
\item $\sfA = \dec_0 = O \left( \frac{d^{3.57} \ln^{6.26}{n}}{\mathsf{W}^{6.26}(d \ln{n})} \right) + O \left( \frac{d^6 \ln^4{n}}{\mathsf{W}^4(d \ln{n})} \right).$
\item $\sfB = O \left(\frac{d^3 \ln^3{n}}{d^2 \mathsf{W}^3(d\ln{n})} \right)$.
\end{itemize}
Because $u\sfB < A$, the decoding complexity is $O(h \times (\sfA + u\sfB) = O(hA) = O \left( h \dec_0 \right)$. The number of tests $t$ is $h(2k+1) = O\left( \frac{d_0^2 \ln{\frac{n}{d_0}}}{(1-p)^2} \times \frac{d^2\ln^2{n}}{\mathsf{W}^2(d\ln{n})} \right)$.

\end{proof}

To reduce the decoding complexity, matrix $\cM$ is chosen as an $(d+1)$-disjunct matrix in Theorem~\ref{thr:mainNonrandom2}. In this case, the number of tests would be increased.

\begin{corollary}
\label{cor:threshold:2}
Let $2 \leq u \leq d < n$ be integers, $d_0 = \max \{u, d-u \}$, and some $p \in [0, 1)$. There exits a $t \times n$ measurement matrix $\cT$ such that up to $d$ defective items in $u$-NATGT can be identified in time $O \left( h \times \max \left\{ \frac{d^2 \ln^3{n}}{\mathsf{W}^2(d\ln{n})}, \frac{ud \ln^4{n}}{\mathsf{W}^3(d\ln{n})} \right\} \right)$ in the presence of up to $e = \Omega(ph/d_0)$ erroneous outcomes, where $t = O\left( h \times \frac{d^2\ln^3{n}}{\mathsf{W}^2(d\ln{n})} \right)$ and $h = O(d_0^2 \left( \ln{\frac{n}{d_0}} \right)/(1 - p)^2).$
\end{corollary}

\begin{proof}
Since matrix $\cM$ is a $k \times n$ $(d+1)$-disjunct matrix in Theorem~\ref{thr:mainNonrandom2}, we have:
\begin{itemize}
\item $k = O \left(\frac{d^2 \ln^3{n}}{\mathsf{W}^2(d\ln{n})} \right)$.
\item $\sfA = O \left(\frac{d^2 \ln^3{n}}{\mathsf{W}^2(d\ln{n})} \right)$.
\item $\sfB = O \left(\frac{d \ln^4{n}}{\mathsf{W}^3(d\ln{n})} \right)$.
\end{itemize}
The decoding complexity is $O(h \times (\sfA + u\sfB) = O(h \max\{ \sfA, u\sfB \} )$. The number of tests $t$ is $h(2k+1) = O\left( \frac{d_0^2 \ln{\frac{n}{d_0}}}{(1-p)^2} \times \frac{d^2\ln^3{n}}{\mathsf{W}^2(d\ln{n})} \right)$.
\end{proof}

\section{Conclusion}
\label{sec:cls}
We have improved encoding and decoding procedures for non-adaptive threshold group testing. The number of tests and the decoding complexity are low in our proposed scheme. Moreover, error tolerance is also considered. However, the proposed scheme works only for $g = 0.$ Therefore, extending the results for $g > 0$ should be studied in future work.

\bibliographystyle{ieeetr}
\bibliography{bibli}

\end{document}